\documentclass[12pt]{article}
\usepackage[utf8]{inputenc}
\usepackage[T1]{fontenc}
\usepackage{array}
\newcolumntype{C}[1]{>{\centering\arraybackslash}m{#1}}
\usepackage{amsmath, amsthm}
\usepackage{amssymb}

\newtheorem{theorem}{Theorem}[section]
\newtheorem{proposition}[theorem]{Proposition}
\newtheorem{corollary}[theorem]{Corollary}
\newtheorem{lemma}[theorem]{Lemma}

\title{Transition Property for $\alpha$-Power Free Languages with $\alpha\geq 2$ and $k\geq 3$ Letters} 

\author{Josef Rukavicka\thanks{Department of Mathematics,
Faculty of Nuclear Sciences and Physical Engineering, CZECH TECHNICAL UNIVERSITY
IN PRAGUE
(josef.rukavicka@seznam.cz).}}

\newtheorem{definition}[theorem]{Definition}
\theoremstyle{remark}

\newtheorem{remark}[theorem]{Remark}

\DeclareMathOperator{\Factor}{F}
\DeclareMathOperator{\pfl}{L}
\DeclareMathOperator{\Prefix}{Prf}
\DeclareMathOperator{\Suffix}{Suf}
\DeclareMathOperator{\occur}{occur}
\DeclareMathOperator{\rext}{rext}
\DeclareMathOperator{\lext}{lext}
\DeclareMathOperator{\RT}{RT}

\date{\small{January 07, 2020}}

\begin{document}
\maketitle

\begin{abstract}
In 1985, Restivo and Salemi presented a list of five problems concerning power free languages. Problem $4$ states: Given $\alpha$-power-free words $u$ and $v$, decide whether there is a transition from $u$ to $v$. Problem $5$ states: Given $\alpha$-power-free words $u$ and $v$, find a transition word $w$, if it exists.

Let $\Sigma_k$ denote an alphabet with $k$ letters. Let $L_{k,\alpha}$ denote the $\alpha$-power free language over the alphabet $\Sigma_k$, where $\alpha$ is a rational number or a rational ``number with $+$''. If $\alpha$ is a  ``number with $+$'' then suppose $k\geq 3$ and $\alpha\geq 2$.
If $\alpha$ is ``only'' a number then suppose $k=3$ and $\alpha>2$ or $k>3$ and $\alpha\geq 2$. We show that: If $u\in L_{k,\alpha}$ is a right extendable word in $L_{k,\alpha}$ and $v\in L_{k,\alpha}$ is a left extendable word in $L_{k,\alpha}$ then there is a (transition) word $w$ such that $uwv\in L_{k,\alpha}$. We also show a construction of the word $w$.
\end{abstract}

\section{Introduction}
The power free words are one of the major themes in the area of combinatorics on words. An $\alpha$-\emph{power} of a word $r$ is the word $r^{\alpha}=rr\dots rt$ such that $\frac{\vert r^{\alpha}\vert}{\vert r\vert}=\alpha$ and $t$ is a prefix of $r$, where $\alpha\geq 1$ is a rational number . For example $(1234)^{3}=123412341234$ and $(1234)^{\frac{7}{4}}=1234123$.
We say that a finite or infinite word $w$ is $\alpha$-\emph{power free} if $w$ has no factors that are $\beta$-powers for $\beta\geq \alpha$
and we say that a finite or infinite word $w$ is $\alpha^+$-power free if $w$ has no factors that are $\beta$-powers for $\beta>\alpha$, where $\alpha, \beta\geq 1$ are rational  numbers.  In the following, when we write ``$\alpha$-power free'' then $\alpha$ denotes a number or a ``number with $+$''. The power free words, also called repetitions free words, include well known square free ($2$-power free), overlap free ($2^+$-power free), and cube free words ($3$-power free). Two surveys on the topic of power free words can be found in \cite{Rampersad_Narad2007} and \cite{10.1007/978-3-642-22321-1_3}.

One of the questions being researched is the construction of infinite power free words. We define the \emph{repetition threshold} $\RT(k)$ to be the infimum of all rational  numbers $\alpha$ such that
there exists an infinite $\alpha$-power-free word over an alphabet with $k$ letters. Dejean's conjecture states that $\RT(2)=2$, $\RT(3)=\frac{7}{4}$, $\RT(4)=\frac{7}{5}$, and $\RT(k)=\frac{k}{k-1}$ for each $k>4$ \cite{DEJEAN197290}. Dejean's conjecture has been proved by the work of several authors \cite{CARPI2007137,Rampersad2011_dejean,DEJEAN197290,OLLAGNIER1992187,PANSIOT1984297,Rao2011_dejean}.

It is easy to see that $\alpha$-power free words form a factorial language \cite{10.1007/978-3-642-22321-1_3}; it means that all factors of a $\alpha$-power free word are also $\alpha$-power free words. Then Dejean's conjecture implies that there are infinitely many finite $\alpha$-power free words over $\Sigma_k$, where $\alpha>\RT(k)$.

In \cite{10.1007/978-3-642-82456-2_20}, Restivo and Salemi presented a list of five problems  that deal with the question of extendability of power free words. In the current paper we investigate Problem $4$ and Problem $5$:
\begin{itemize} 
\item
Problem $4$:  Given $\alpha$-power-free words $u$ and $v$, decide whether there is a transition word $w$, such that $uwu$ is $\alpha$-power free.
\item
Problem $5$: Given $\alpha$-power-free words $u$ and $v$, find a transition word $w$, if it exists.
\end{itemize}
A recent survey on the progress of solving all the five problems can be found in \cite{10.1007/978-3-030-19955-5_27}; in particular, the problems $4$ and $5$ are solved for some overlap free ($2^+$-power free) binary words. In addition, in \cite{10.1007/978-3-030-19955-5_27} the authors prove that: For every pair $(u,v)$ of cube free words ($3$-power free) over an alphabet with $k$ letters, if $u$ can be infinitely extended to the right and $v$ can be infinitely extended to the left respecting the cube-freeness property, then there exists a “transition” word $w$ over the same alphabet such that $uwv$ is cube free.

In 2009, a conjecture related to Problems $4$ and Problem $5$ of Restivo and Salemi appeared in \cite{10.1007/978-3-642-02737-6_38}: Let $L$ be a power-free language, $u, v \in e(L)$, where $e(L)\subseteq L$ is the set of extendable words of $L$. Then $uwv \in e(L)$ for some word $w$. In 2018, this conjecture was presented also in \cite{SHALLIT201996} in a slightly different form.

Let $\mathbb{N}$ denote the set of natural numbers and let $\mathbb{Q}$ denote the set of rational numbers.
\begin{definition}\label{rrnj221g5} Let 
\[
\begin{split}
\Upsilon=
\{(k,\alpha)\mid k\in \mathbb{N}\mbox{ and }\alpha\in \mathbb{Q}\mbox{ and }k=3 \mbox{ and }\alpha>2\}\\ \cup\{(k,\alpha)\mid k\in \mathbb{N}\mbox{ and }\alpha\in \mathbb{Q}\mbox{ and }k>3\mbox{ and }\alpha\geq 2\}\\ \cup\{(k,\alpha^+)\mid k\in \mathbb{N}\mbox{ and }\alpha\in \mathbb{Q}\mbox{ and } k\geq3\mbox{ and }\alpha\geq 2\}\mbox{.}
\end{split}
\]
\end{definition}
\begin{remark}
The definition of $\Upsilon$ says that:
If $(k,\alpha)\in \Upsilon$ and $\alpha$ is a ``number with $+$'' then $k\geq 3$ and $\alpha\geq 2$.
If $(k,\alpha)\in \Upsilon$ and $\alpha$ is ``just'' a number then $k=3$ and $\alpha>2$ or $k>3$ and $\alpha\geq 2$.
\end{remark}

Let $\pfl$ be a language. A finite word $w\in \pfl$ is called \emph{left extendable}  (resp., \emph{right extendable}) in $\pfl$ if for every $n\in \mathbb{N}$ there is a word $u\in \pfl$ with $\vert u\vert=n$ such that $uw\in \pfl$ (resp., $wu\in \pfl$).

In the current article we improve the results addressing Problems $4$ and Problem $5$ of Restivo and Salemi from \cite{10.1007/978-3-030-19955-5_27} as follows. Let $\Sigma_k$ denote an alphabet with $k$ letters. Let $\pfl_{k,\alpha}$ denote the $\alpha$-power free language over the alphabet $\Sigma_k$. We show that if $(k,\alpha)\in \Upsilon$, $u\in \pfl_{k,\alpha}$ is a right extendable word in $\pfl_{k,\alpha}$, and $v\in \pfl_{k,\alpha}$ is a left extendable word in $\pfl_{k,\alpha}$ then there is a word $w$ such that $uwv\in \pfl_{k,\alpha}$. We also show a construction of the word $w$.

We sketch briefly our construction of a ``transition'' word. Let $u$ be a right extendable $\alpha$-power free word and let $v$ be a left extendable $\alpha$-power free word over $\Sigma_k$ with $k>2$ letters. Let $\bar u$ be a right infinite $\alpha$-power free word having $u$ as a prefix and let $\bar v$ be a left infinite $\alpha$-power free word having $v$ as a suffix. Let $x$ be a letter that is recurrent in both $\bar u$ and $\bar v$. We show that we may suppose that $\bar u$ and $\bar v$ have a common recurrent letter. Let $t$ be a right infinite $\alpha$-power free word over $\Sigma_{k}\setminus\{x\}$. Let $\bar t$ be a left infinite $\alpha$-power free word such that the set of factors of $\bar t$ is a subset of the set of recurrent factors of $t$. We show that such $\bar t$ exists. We identify a prefix $\tilde uxg$ of $\bar u$ such that $g$ is a prefix of $t$ and $\tilde uxt$ is a right infinite $\alpha$-power free word. Analogously we identify a suffix $\bar gx\tilde v$ of $\bar v$ such that $\bar g$ is a suffix of $\bar t$ and $\bar tx\tilde v$ is a left infinite $\alpha$-power free word. Moreover our construction guarantees that $u$ is a prefix of $\tilde uxt$ and $v$ is a suffix of $\bar tx\tilde v$. 
Then we find a prefix $hp$ of $t$ such that $px\tilde v$ is a suffix of $\bar tx\tilde v$ and such that both $h$ and $p$ are ``sufficiently long''. Then we show that $\tilde uxhpx\tilde v$ is an $\alpha$-power free word having $u$ as a prefix and $v$ as a suffix.

The very basic idea of our proof is that if $u,v$ are $\alpha$-power free words and $x$ is a letter such that $x$ is not a factor of both $u$ and $v$, then clearly $uxv$ is $\alpha$-power free  on condition that $\alpha\geq 2$. Just note that there cannot be a factor in $uxv$ which is an $\alpha$-power and contains $x$, because $x$ has only one occurrence in $uxv$.
Our constructed words $\tilde uxt$, $\bar tx\tilde v$, and $\tilde uxhpx\tilde v$ have ``long'' factors which does not contain a letter $x$. This will allow us to apply a similar approach to show that the constructed words does not contain square factor $rr$ such that $r$ contains the letter $x$.

Another key observation is that if $k\geq 3$ and $\alpha>\RT(k-1)$ then there is an infinite $\alpha$-power free word $\bar w$ over $\Sigma_k\setminus\{x\}$, where $x\in \Sigma_k$. This is an implication of Dejean's conjecture. Less formally said, if $u,v$ are $\alpha$-power free words over an alphabet with $k$ letters, then we construct a ``transition'' word $w$ over an alphabet with $k-1$ letters  such that $uwv$ is $\alpha$-power free. 

Dejean's conjecture imposes also the limit to possible improvement of our construction. The construction cannot be used for $\RT(k)\leq \alpha<\RT(k-1)$, where $k\geq 3$, because every infinite (or ``sufficiently long'') word $w$ over an alphabet with $k-1$ letters contains a factor which is an $\alpha$-power. Also for $k=2$ and $\alpha\geq 1$ our technique fails. On the other hand, based on our research, it seems that our technique, with some adjustments, could be applied also for $\RT(k-1)\leq\alpha\leq 2$ and $k\geq 3$. 

\section{Preliminaries}
Recall that $\Sigma_k$ denotes an alphabet with $k$ letters. Let $\epsilon$ denote the empty word.
Let $\Sigma_k^*$ denote the set of all finite words over $\Sigma_k$ including the empty word $\epsilon$, let $\Sigma_k^{\mathbb{N},R}$ denote the set of all right infinite words over $\Sigma_k$, and let $\Sigma_k^{\mathbb{N},L}$ denote the set of all left infinite words over $\Sigma_k$. Let $\Sigma_k^{\mathbb{N}}=\Sigma_k^{\mathbb{N},L}\cup \Sigma_k^{\mathbb{N},R}$. We call $w\in \Sigma_k^{\mathbb{N}}$ an infinite word.

Let $\occur(w,t)$ denote the number of occurrences of the nonempty factor $t\in \Sigma_k^*\setminus\{\epsilon\}$ in the word $w\in \Sigma_k^*\cup\Sigma_k^{\mathbb{N}}$. If $w\in \Sigma_k^{\mathbb{N}}$ and $\occur(w,t)=\infty$, then we call $t$ a \emph{recurrent} factor in $w$.

Let $\Factor(w)$ denote the set of all finite factors of a finite or infinite word $w\in \Sigma_k^*\cup\Sigma_k^{\mathbb{N}}$. The set $\Factor(w)$ contains the empty word and if $w$ is finite then also $w\in \Factor(w)$. Let $\Factor_r(w)\subseteq \Factor(w)$ denote the set of all recurrent nonempty factors of $w\in \Sigma_k^{\mathbb{N}}$.

Let $\Prefix(w)\subseteq \Factor(w)$ denote the set of all prefixes of $w\in \Sigma_k^*\cup\Sigma_k^{\mathbb{N},R}$ and let $\Suffix(w)\subseteq \Factor(w)$ denote the set of all suffixes of $w\in \Sigma_k^*\cup\Sigma_k^{\mathbb{N},L}$. We define that $\epsilon\in \Prefix(w)\cap\Suffix(w)$ and if $w$ is finite then also $w\in \Prefix(w)\cap\Suffix(w)$. 

We have that $\pfl_{k,\alpha}\subseteq\Sigma_k^*$. Let $\pfl_{k,\alpha}^{\mathbb{N}}\subseteq\Sigma_k^{\mathbb{N}}$ denote the set of all infinite $\alpha$-power free words over $\Sigma_k$. Obviously $\pfl_{k,\alpha}^{\mathbb{N}}=\{w\in \Sigma_k^{\mathbb{N}}\mid\Factor(w)\subseteq \pfl_{k,\alpha}\}$. In addition we define $\pfl_{k,\alpha}^{\mathbb{N},R}=\pfl_{k,\alpha}^{\mathbb{N}}\cap\Sigma_k^{\mathbb{N},R}$ and $\pfl_{k,\alpha}^{\mathbb{N},L}=\pfl_{k,\alpha}^{\mathbb{N}}\cap\Sigma_k^{\mathbb{N},L}$; it means the sets of right infinite and left infinite $\alpha$-power free words.

\section{Power Free Languages}

Let $(k,\alpha)\in \Upsilon$ and let $u,v$ be $\alpha$-power free words. The first lemma says that $uv$ is $\alpha$-power free if there are no word $r$ and no nonempty prefix $\bar v$ of $v$ such that $rr$ is a suffix of $u\bar v$ and $rr$ is longer than $\bar v$. 
\begin{lemma}
\label{tnh900u0}
Suppose $(k,\alpha)\in \Upsilon$, $u\in \pfl_{k,\alpha}$, and $v\in\pfl_{k,\alpha}\cup \pfl_{k,\alpha}^{\mathbb{N},R}$. Let 
\[
\begin{split}\Pi=\{(r,\bar v)\mid r\in \Sigma_k^*\setminus\{\epsilon\}\mbox{ and }\bar v\in \Prefix(v)\setminus\{\epsilon\}\mbox{ and } \\ 
rr\in \Suffix(u\bar v)\mbox{ and }\vert rr\vert>\vert \bar v\vert\}\mbox{.}\end{split}\] If $\Pi=\emptyset$ then $uv\in \pfl_{k,\alpha}\cup \pfl_{k,\alpha}^{\mathbb{N},R}$.
\end{lemma}
\begin{proof}
Suppose that $uv$ is not $\alpha$-power free. Since $u$ is $\alpha$-power free, then there are $t\in \Sigma_k^*$ and $x\in \Sigma_k$ such that $tx\in \Prefix(v)$, $ut\in \pfl_{k,\alpha}$ and $utx\not \in \pfl_{k,\alpha}$. 
It means that there is $r\in \Suffix(utx)$ such that $r^{\beta}\in \Suffix(utx)$ for some $\beta\geq \alpha$ or $\beta>\alpha$ if $\alpha$ is a ``number with $+$''; recall Definition \ref{rrnj221g5} of $\Upsilon$. Because $\alpha\geq 2$, this implies that $rr\in \Suffix(r^{\beta})$. If follows that $(tx,r)\in \Pi$. We proved that $uv\not\in \pfl_{k,\alpha}\cup \pfl_{k,\alpha}^{\mathbb{N},R}$ implies that $\Pi\not=\emptyset$. The lemma follows.
\end{proof}

The following technical set $\Gamma(k,\alpha)$ of $5$-tuples $(w_1,w_2,x,g,t)$ will simplify our propositions.
\begin{definition}
\label{rhtuj66d85tr89}
Given $(k,\alpha)\in \Upsilon$, we define that $(w_1,w_2,x,g,t)\in \Gamma(k,\alpha)$ if
\begin{enumerate}
\item
\label{rjfurkoprp01}
$w_1,w_2,g\in \Sigma_k^*$, 
\item
\label{rjfurkoprp02}
$x\in \Sigma_k$,
\item
\label{rjfurkoprp03}
$w_1w_2xg\in \pfl_{k,\alpha}$,
\item
\label{rjfurkoprp04}
$t\in \pfl_{k,\alpha}^{\mathbb{N,R}}$,
\item
\label{rjfurkoprp05}
$\occur(t,x)=0$,
\item
\label{rjfurkoprp06}
$g\in \Prefix(t)$,
\item
\label{tngjh5d5c26g5}
$\occur(w_2xgy,xgy)=1$, where $y\in \Sigma_k$ is such that $gy\in \Prefix(t)$, and
\item
\label{yy7utiuyjhki8}
$\occur(w_2,x)\geq\occur(w_1,x)$.
\end{enumerate}
\end{definition} 
\begin{remark}
Less formally said, the $5$-tuple $(w_1,w_2,x,g,t)$ is in $\Gamma(k,\alpha)$ if $w_1w_2xg$ is $\alpha$-power free word over $\Sigma_k$, $t$ is a right infinite $\alpha$-power free word over $\Sigma_k$, $t$ has no occurrence of $x$ (thus $t$ is a word over $\Sigma_k\setminus\{x\}$), $g$ is a prefix of $t$, $xgy$ has only one occurrence in $w_2xgy$, where $y$ is a letter such that $gy$ is a prefix of $t$, and the number of occurrences of $x$ in $w_2$ is bigger than the number of occurrences of $x$ in $w_1$, where $w_1,w_2, g$ are finite words and $x$ is a letter.
\end{remark}

The next proposition shows that if $(w_1,w_2,x,g,t)$ is from the set $\Gamma(k,\alpha)$ then $w_1w_2xt$ is a right infinite $\alpha$-power free word, where $(k,\alpha)$ is from the set $\Upsilon$. 
\begin{proposition}
\label{tnmbvn21215fhg}
If $(k,\alpha)\in \Upsilon$ and $(w_1,w_2,x,g,t)\in \Gamma(k,\alpha)$ then $w_1w_2xt\in \pfl_{k,\alpha}^{\mathbb{N},R}$.
\end{proposition}
\begin{proof}
Lemma \ref{tnh900u0} implies that it suffices to show that there are no $u\in \Prefix(t)$ with $\vert u\vert>\vert g\vert$ and no $r\in \Sigma_k^*\setminus \{\epsilon\}$ such that $rr\in \Suffix(w_1w_2xu)$ and $\vert rr\vert>\vert u\vert$. Recall that $w_1w_2xg$ is an $\alpha$-power free word, hence we consider $\vert u\vert >\vert g\vert$. To get a contradiction, suppose that such $r,u$ exist. We distinguish following distinct cases.
\begin{itemize}
\item
If $\vert r\vert\leq \vert u\vert$ then: Since $u\in\Prefix(t)\subseteq \pfl_{k,\alpha}$ it follows that $xu\in \Suffix(r^2)$ and hence $x\in \Factor(r^2)$. It is clear that $\occur(r^2,x)\geq 1$ if and only if $\occur(r,x)\geq1$. Since $x\not \in \Factor(u)$ and thus $x\not \in \Factor(r)$, this is a contradiction.
\item
If $\vert r\vert>\vert u\vert$ and $rr\in \Suffix(w_2xu)$ then: Let $y\in \Sigma_k$ be such that $gy\in \Prefix(t)$. Since $\vert u\vert>\vert g\vert$ we have that $gy\in \Prefix(u)$ and $xgy\in \Prefix(xu)$. Since $\vert r\vert>\vert u\vert$ we have that $xgy\in \Factor(r)$. In consequence $\occur(rr,xgy)\geq 2$. 

But Property \ref{tngjh5d5c26g5} of Definition \ref{rhtuj66d85tr89} states that $\occur(w_2xgy, xgy)=1$. Since $rr\in \Suffix(w_2xu)$, this is a contradiction.
\item
If $\vert r\vert>\vert u\vert$ and $rr\not \in \Suffix(w_2xu)$ and $r\in \Suffix(w_2xu)$ then: 

Let $w_{11}, w_{12}, w_{13}, w_{21}, w_{22}\in \Sigma_k^*$ be such that $w_1=w_{11}w_{12}w_{13}$,  $w_2=w_{21}w_{22}$, $w_{12}w_{13}w_{21}=r$, $w_{12}w_{13}w_{2}xu=rr$, and $w_{13}w_{21}=xu$; see Figure below.
\begin{table}[ht]
\centering
\begin{tabular}{ll|l|l|lll}
\cline{3-4}
                               &          & \multicolumn{2}{c|}{$xu$} &                               &                          &                          \\ \hline
\multicolumn{1}{|l|}{$w_{11}$} & $w_{12}$ & $w_{13}$    & $w_{21}$    & \multicolumn{1}{l|}{$w_{22}$} & \multicolumn{1}{l|}{$x$} & \multicolumn{1}{l|}{$u$} \\ \hline
\multicolumn{1}{|l|}{}         & \multicolumn{3}{c|}{$r$}             & \multicolumn{3}{c|}{$r$}                                                            \\ \hline
\end{tabular}
\end{table}

It follows that $w_{22}xu=r$ and $w_{22}=w_{12}$. 
It is easy to see that $w_{13}w_{21}=xu$. For $\occur(u,x)=0$ we have that $\occur(w_2,x)=\occur(w_{22},x)$ and $\occur(w_{13},x)=1$. For $w_{22}=w_{12}$ it follows that $\occur(w_1,x)>\occur(w_2,x)$. This is a contradiction to Property \ref{yy7utiuyjhki8} of Definition \ref{rhtuj66d85tr89}.
\item
If $\vert r\vert>\vert u\vert$ and $rr\not \in \Suffix(w_2xu)$ and $r\not \in \Suffix(w_2xu)$ then: 

Let $w_{11}, w_{12}, w_{13}\in \Sigma_k^*$ be such that $w_1=w_{11}w_{12}w_{13}$, $w_{12}=r$ and $w_{13}w_2xu=r$; see Figure below.
\begin{table}[ht]
\centering
\begin{tabular}{|c|C{2cm}|c|c|c|c|}
\hline
$w_{11}$ & $w_{12}$                 & $w_{13}$ & $w_{2}$ & $x$ & $u$ \\ \hline
         & \multicolumn{1}{c|}{$r$} & \multicolumn{4}{c|}{$r$}       \\ \hline
\end{tabular}
\end{table}

It follows that \[\occur(w_{12},x)=\occur(w_{13},x)+\occur(w_2,x)+\occur(xu,x)\mbox{.}\] This is a contradiction to Property \ref{yy7utiuyjhki8} of Definition \ref{rhtuj66d85tr89}.
\end{itemize}
We proved that the assumption of existence of $r,u$ leads to a contradiction. Thus we proved that for each prefix $u\in \Prefix(t)$ we have that $w_1w_2xu\in \pfl_{k,\alpha}$. The proposition follows.
\end{proof}
We prove that if $(k,\alpha)\in \Upsilon$ then there is a right infinite $\alpha$-power free word over $\Sigma_{k-1}$. In the introduction we showed that this observation could be deduced from Dejean's conjecture. Here additionally, to be able to address Problem $5$ from the list of Restivo and Salemi, we present in the proof also examples of such words. 
\begin{lemma}
\label{htjnd8738kfhj38}
If $(k,\alpha)\in \Upsilon$ then the set $\pfl_{k-1,\alpha}^{\mathbb{N},R}$ is not empty.
\end{lemma}
\begin{proof}
If $k=3$ then $\vert \Sigma_{k-1}\vert = 2$. It is well known that the Thue Morse word is a right infinite $2^+$-power free word over an alphabet with $2$ letters \cite{SHALLIT201996}.  It follows that the Thue Morse word is $\alpha$-power free for each $\alpha>2$.

If $k>3$ then $\vert \Sigma_{k-1}\vert \geq 3$. It is well known that there are infinite $2$-power free words over an alphabet with $3$ letters \cite{SHALLIT201996}. Suppose $0,1,2\in\Sigma_k$. An example is the fixed point of the morphism $\theta$ defined by $\theta(0)=012$, $\theta(1)=02$, and $\theta(2)=1$ \cite{SHALLIT201996}. If an infinite word $t$ is $2$-power free then obviously $t$ is $\alpha$-power free and $\alpha^+$-power free for each $\alpha\geq 2$.

This completes the proof.
\end{proof}

We define the sets of extendable words.
\begin{definition}
Let $\pfl\subseteq \Sigma_k^*$. We define 
\[\lext(\pfl)=\{w\in \pfl\mid w\mbox{ is left extendable in }\pfl\}\] and 
\[\rext(\pfl)=\{w\in \pfl\mid w\mbox{ is right extendable in }\pfl\}\mbox{.}\]
If $u\in \lext(\pfl)$ then let $\lext(u,\pfl)$ be the set of all left infinite words $\bar u$ such that $\Suffix(\bar u)\subseteq \pfl$ and $u\in \Suffix(\bar u)$. 
Analogously if  $u\in \rext(\pfl)$ then let $\rext(u,\pfl)$ be the set of all right infinite words $\bar u$ such that $\Prefix(\bar u)\subseteq \pfl$ and $u\in \Prefix(\bar u)$. 
\end{definition}
We show the sets $\lext(u,\pfl)$ and $\rext(v,\pfl)$ are nonempty for left extendable and rigth extendable words.
\begin{lemma}
\label{yh567cbng6}
If $\pfl\subseteq \Sigma_k^*$ and $u\in \lext(\pfl)$ (resp., $v\in \rext(\pfl)$) then \\ $\lext(u,\pfl)\not=\emptyset$ (resp., $\rext(v,\pfl)\not=\emptyset$).
\end{lemma}
\begin{proof}
Realize that $u\in \lext(\pfl)$ (resp., $v\in \rext(\pfl)$) implies that there are infinitely many finite words in $\pfl$ having $u$ as a suffix (resp., $v$ as a prefix).
Then the lemma follows from König's Infinity Lemma \cite{Koenig1926}, \cite{Rampersad_Narad2007}.
\end{proof}
The next proposition proves that if $(k,\alpha)\in \Upsilon$, $w$ is a right extendable $\alpha$-power free word, $\bar w$ is a right infinite $\alpha$-power free word having the letter $x$ as a recurrent factor and having $w$ as a prefix, and $t$ is a right infinite $\alpha$-power free word over $\Sigma_k\setminus\{x\}$, then there are finite words $w_1,w_2,g$ such that the $5$-tuple $(w_1,w_2,x,g,t)$ is in the set $\Gamma(k,\alpha)$ and $w$ is a prefix of $w_1w_2xg$. 
\begin{proposition}
\label{tjn9epl22l}
If $(k,\alpha)\in \Upsilon$, $w\in \rext(\pfl_{k,\alpha})$, $\bar w\in \rext(w,\pfl_{k,\alpha})$, $x\in \Factor_r(\bar w)\cap\Sigma_k$, $t\in \pfl_{k,\alpha}^{\mathbb{N},R}$, and $\occur(t,x)=0$ then there are finite words $w_1,w_2,g$ such that $(w_1,w_2,x,g,t)\in \Gamma(k,\alpha)$ and $w\in \Prefix(w_1w_2xg)$.
\end{proposition}
\begin{proof}
Let $\omega=\Factor(\bar w)\cap\Prefix(xt)$ be the set of factors of $\bar w$ that are also prefix of the word $xt$. Based on the size of the set $\omega$ we construct the words $w_1,w_2,g$ and we show that $(w_1,w_2,x,g,t)\in \Gamma(k,\alpha)$ and $w_1w_2xg\in \Prefix(\bar w)\subseteq \pfl_{k,\alpha}$. The properties \ref{rjfurkoprp01}, \ref{rjfurkoprp02}, \ref{rjfurkoprp03}, \ref{rjfurkoprp04}, \ref{rjfurkoprp05}, and \ref{rjfurkoprp06} of Definition \ref{rhtuj66d85tr89} are easy to verify. Hence we explicitely prove only properties \ref{tngjh5d5c26g5} and \ref{yy7utiuyjhki8} and that $w\in \Prefix(w_1w_2xg)$.
\begin{itemize}
\item
If $\omega$ is an infinite set. It follows that $\Prefix(xt)=\omega$. Let $g\in \Prefix(t)$ be such that $\vert g\vert=\vert w\vert$ and let $w_2\in \Prefix(\bar w)$ be such that $w_2xg \in \Prefix(\bar w)$ and $\occur(w_2xg, xg)=1$. Let $w_1=\epsilon$. 

Property \ref{tngjh5d5c26g5} of Definition \ref{rhtuj66d85tr89} follows from $\occur(w_2xg, xg)=1$.
Property \ref{yy7utiuyjhki8} of Definition \ref{rhtuj66d85tr89} is obvious, because $w_1$ is the empty word. Since $\vert g\vert=\vert w\vert$ and $w\in \Prefix(\bar w)$ we have that $w\in \Prefix(w_1w_2xg)$.
\item
If $\omega$ is a finite set. Let $\bar \omega= \omega\cap \Factor_r(\bar w)$ be the set of prefixes of $xt$ that are recurrent in $\bar w$. Since $x$ is recurrent in $\bar w$ we have that $x\in \bar \omega$ and thus $\bar \omega$ is not empty. Let $g\in \Prefix(t)$ be such that $xg$ is the longest element in $\bar\omega$.  Let $w_1\in \Prefix(w)$ be the shortest prefix of $\bar w$ such that if $u\in \omega\setminus\bar \omega$ is a non-recurrent prefix of $xt$ in $\bar w$ then $\occur(w_1,u)=\occur(\bar w,u)$. Such $w_1$ obviously exists, because $\omega$ is a finite set and non-recurrent factors have only a finite number of occurrences.
Let $w_2$ be the shortest factor of $\bar w$ such that $w_1w_2xg\in \Prefix(\bar w)$, $\occur(w_1,x)<\occur(w_2,x)$, and $w\in \Prefix(w_1w_2xg)$. Since $xg$ is recurrent in $\bar w$ and $w\in \Prefix(\bar w)$ it is clear such $w_2$ exists.

We show that Property \ref{tngjh5d5c26g5} of Definition \ref{rhtuj66d85tr89} holds. Let $y\in \Sigma_k$ be such that $gy\in \Prefix(t)$. Suppose that $\occur(w_2xg,xgy)>0$. It would imply that $xgy$ is recurrent in $\bar w$, since all occurrences of non-recurrent words from $\omega$ are in $w_1$. But we defined $xg$ to be the longest recurrent word $\omega$. Hence it is contradiction to our assumption that $\occur(w_2xg,xgy)>0$. 

Property \ref{yy7utiuyjhki8} of Definition \ref{rhtuj66d85tr89} and $w\in \Prefix(w_1w_2xg)$ are obvious from the construction of $w_2$.
\end{itemize}
This completes the proof.
\end{proof}

We define the \emph{reversal} $w^R$ of a finite or infinite word $w=\Sigma_k^*\cup\Sigma_k^{\mathbb{N}}$ as follows: If $w\in\Sigma_k^*$ and $w=w_1w_2\dots w_m$, where $w_i\in \Sigma_k$ and $1\leq i\leq m$, then $w^R=w_mw_{m-1}\dots w_2w_1$. If $w\in \Sigma_k^{\mathbb{N},L}$ and $w=\dots w_2w_1$, where $w_i\in \Sigma_k$ and $i\in \mathbb{N}$, then $w^R=w_1w_2\dots\in \Sigma_k^{\mathbb{N},R}$. Analogously if $w\in \Sigma_k^{\mathbb{N},R}$ and $w=w_1w_2\dots$, where $w_i\in \Sigma_k$ and $i\in \mathbb{N}$, then $w^R=\dots w_2w_1\in \Sigma_k^{\mathbb{N},L}$.

Proposition \ref{tnmbvn21215fhg} allows to construct a right infinite $\alpha$-power free word with a given prefix. The next simple corollary shows that in the same way we can construct a left infinite $\alpha$-power free word with a given suffix.
\begin{corollary}
\label{thyhc7c7cvg}
If $(k,\alpha)\in \Upsilon$, $w\in \lext(\pfl_{k,\alpha})$, $\bar w\in \lext(w,\pfl_{k,\alpha})$, $x\in \Factor_r(\bar w)\cap\Sigma_k$, $t\in \pfl_{k,\alpha}^{\mathbb{N},L}$, and $\occur(t,x)=0$ then there are finite words $w_1,w_2,g$ such that $(w_1^R,w_2^R,x,g^R,t^R)\in \Gamma(k,\alpha)$, $w\in \Suffix(gxw_2w_1)$, and $txw_2w_1\in \pfl_{k,\alpha}^{\mathbb{N},L}$.
\end{corollary}
\begin{proof}
Let $u\in \Sigma_k^*\cup \Sigma_k^{\mathbb{N}}$. Realize that $u\in \pfl_{k,\alpha}\cup\pfl_{k,\alpha}^{\mathbb{N}}$ if and only if $u^R\in \pfl_{k,\alpha}\cup\pfl_{k,\alpha}^{\mathbb{N}}$. Then the corollary follows from Proposition \ref{tnmbvn21215fhg} and Proposition \ref{tjn9epl22l}.
\end{proof}

Given $k\in\mathbb{N}$ and a right infinite word $t\in\Sigma_k^{\mathbb{N},R}$, let $\Phi(t)$ be the set of all left infinite words $\bar t\in \Sigma_k^{\mathbb{N},L}$ such that $\Factor(\bar t)\subseteq \Factor_r(t)$. It means that all factors of $\bar t\in \Phi(t)$ are recurrent factors of $t$. We show that the set $\Phi(t)$ is not empty.
\begin{lemma}
\label{lrktikl009iu8}
If $k\in\mathbb{N}$ and $t\in \Sigma_k^{\mathbb{N},R}$ then $\Phi(t)\not=\emptyset$.
\end{lemma}
\begin{proof}
Since $t$ is an infinite word, the set of recurrent factors of $t$ is not empty. 
Let $g$ be a recurrent nonempty factor of $t$; $g$ may be a letter. Obviously there is $x\in \Sigma_k$ such that $xg$ is also recurrent in $t$. This implies that the set $\{h\mid hg\in\Factor_r(t)\}$ is infinite. The lemma follows from König's Infinity Lemma \cite{Koenig1926}, \cite{Rampersad_Narad2007}.
\end{proof}

The next lemma shows that if $u$ is a right extendable $\alpha$-power free word then for each letter $x$ there is a right infinite $\alpha$-power free word $\bar u$ such that $x$ is recurrent in $\bar u$ and $u$ is a prefix of $\bar u$.
\begin{lemma}
\label{tjh999g0hkgo}
If $(k,\alpha)\in\Upsilon$, $u\in \rext(\pfl_{k,\alpha})$, and $x\in \Sigma_{k}$ then  there is $\bar u\in \rext(u,\pfl_{k,\alpha})$ such that $x\in \Factor_r(\bar u)$.
\end{lemma}
\begin{proof}
Let $w\in \rext(u,\pfl_{k,\alpha})$; Lemma \ref{yh567cbng6} implies that $\rext(u,\pfl_{k,\alpha})$ is not empty. If $x\in\Factor_r(w)$ then we are done. Suppose that $x\not\in\Factor_r(w)$. Let $y\in \Factor_r(w)\cap\Sigma_k$. Clearly $x\not=y$. Proposition \ref{tjn9epl22l} implies that there is $(w_1,w_2,y,g,t)\in \Gamma(k,\alpha)$ such that $u\in \Prefix(w_1w_2yg)$. The proof of Lemma \ref{htjnd8738kfhj38} implies that we can choose $t$ in such a way that $x$ is recurrent in $t$. Then $w_1w_2yt\in \rext(u,\pfl_{k,\alpha})$ and $x\in \Factor_r(w_1w_2yt)$.
This completes the proof.
\end{proof}

The next proposition shows that if $u$ is left extendable and $v$ is right extendable then there are finite words $\tilde u, \tilde v$, a letter $x$, a right infinite word $t$, and a left infinite word $\bar t$ such that $\tilde uxt, \bar tx\tilde v$ are infinite $\alpha$-power free words, $t$ has no occurrence of $x$, every factor of $\bar t$ is a recurrent factor in $t$, $u$ is a prefix of $\tilde uxt$, and $v$ is a suffix of $\bar tx\tilde v$.
\begin{proposition}
\label{u8bnzxhg99}
If $(k,\alpha)\in \Upsilon$, $u\in \rext(\pfl_{k,\alpha})$, and $v\in \lext(\pfl_{k,\alpha})$ then there are $\tilde u, \tilde v\in \Sigma_k^*$, $x\in \Sigma_k$, $t\in \Sigma_k^{\mathbb{N},R}$, and $\bar t\in \Sigma_k^{\mathbb{N},L}$ such that $\tilde uxt\in \pfl_{k,\alpha}^{\mathbb{N},R}$, $\bar tx\tilde v\in \pfl_{k,\alpha}^{\mathbb{N},L}$, $\occur(t,x)=0$, $\Factor(\bar t)\subseteq \Factor_r(t)$, $u\in \Prefix(\tilde uxt)$, and $v\in \Suffix(\bar tx\tilde v)$.
\end{proposition}
\begin{proof}
Let $\bar u\in \rext(u,\pfl_{k,\alpha})$ and $\bar v\in \lext(v,\pfl_{k,\alpha})$ be such that $\Factor_r(\bar u)\cap\Factor_r(\bar v)\cap \Sigma_k\not =\emptyset$. Lemma \ref{tjh999g0hkgo} implies that such $\bar u, \bar v$ exist. Let $x\in \Factor_r(\bar u)\cap\Factor_r(\bar v)\cap \Sigma_k$. It means that the letter $x$ is recurrent in both $\bar u$ and $\bar v$.

Let $t$ be a right infinite  $\alpha$-power free word over $\Sigma_k\setminus\{x\}$. Lemma \ref{htjnd8738kfhj38} asserts that such $t$ exists. Let $\bar t\in\Phi(t)$; Lemma \ref{lrktikl009iu8} shows that $\Phi(t)\not =\emptyset$. It is easy to see that $\bar t\in \pfl_{k,\alpha}^{\mathbb{N},L}$, because $\Factor(\bar t)\subseteq \Factor_r(t)$ and $t\in \pfl_{k,\alpha}^{\mathbb{N},R}$.

Proposition \ref{tjn9epl22l} and Corollary \ref{thyhc7c7cvg} imply that there are $u_1,u_2,g,v_1,v_2,\bar g\in \pfl_{k,\alpha}$ such that 
\begin{itemize}
\item $(u_1,u_2,x,g,t)\in\Gamma(k,\alpha)$, 
\item $(v_1^R,v_2^R,x,\bar g^R,\bar t^R)\in \Gamma(k,\alpha)$,  
\item $u\in \Prefix(u_1u_2xg)$, and
\item $v^R\in \Prefix(v_1^Rv_2^Rx\bar g^R)$; it follows that $v\in \Suffix(\bar gxv_2v_1)$. 
\end{itemize}
Proposition \ref{tnmbvn21215fhg} implies that $u_1u_2xt, v_1^Rv_2^Rx\bar t^R\in \pfl_{k,\alpha}^{\mathbb{N},R}$. It follows that $\bar txv_2v_1\in \pfl_{k,\alpha}^{\mathbb{N},L}$. Let $\tilde u=u_1u_2$ and $\tilde v=v_2v_1$.
This completes the proof.
\end{proof}
The main theorem of the article shows that if $u$ is a right extendable $\alpha$-power free word and $v$ is a left extendable $\alpha$-power free word then there is a word $w$ such that $uwv$ is $\alpha$-power free. The proof of the theorem shows also a construction of the word $w$.
\begin{theorem}
If $(k,\alpha)\in \Upsilon$, $u\in \rext(\pfl_{k,\alpha})$, and $v\in \lext(\pfl_{k,\alpha})$ then there is $w\in \pfl_{k,\alpha}$ such that $uwv\in \pfl_{k,\alpha}$.
\end{theorem}
\begin{proof}
Let $\tilde u,\tilde v,x,t,\bar t$ be as in Proposition \ref{u8bnzxhg99}.
Let $p\in \Suffix(\bar t)$ be the shortest suffix such that $\vert p\vert>\max\{\vert \tilde ux\vert,\vert x\tilde v\vert, \vert u\vert, \vert v\vert\}$. Let $h\in \Prefix(t)$ be the shortest prefix such that $hp\in \Prefix(t)$ and $\vert h\vert>\vert p\vert$; such $h$ exists, because $p$ is a recurrent factor of $t$; see Proposition \ref{u8bnzxhg99}.
We show that $\tilde uxhpx\tilde v\in \pfl_{k,\alpha}$.

We have that $\tilde uxhp\in \pfl_{k,\alpha}$, since $hp\in \Prefix(t)$ and Proposition \ref{u8bnzxhg99} states that $\tilde uxt\in \pfl_{k,\alpha}^{\mathbb{N},R}$. 
Lemma \ref{tnh900u0} implies that it suffices to show that there are no $g\in \Prefix(\tilde v)$ and no $r\in \Sigma_k^*\setminus \{\epsilon\}$ such that $rr\in \Suffix(\tilde uxhpxg)$ and $\vert rr\vert>\vert xg\vert$. To get a contradiction, suppose there are such $r,g$. We distinguish following cases.
\begin{itemize}
\item
If $\vert r\vert\leq \vert xg\vert$ then $rr\in \Suffix(pxg)$, because $\vert p\vert>\vert x\tilde v\vert$ and $xg\in \Prefix(x\tilde v)$. This is a contradiction, since $px\tilde v\in \Suffix(\bar tx\tilde v)$ and $\bar tx\tilde v\in \pfl_{k,\alpha}^{\mathbb{N},L}$; see Proposition \ref{u8bnzxhg99}.
\item
If $\vert r\vert>\vert xg\vert$ then $\vert r\vert\leq \frac{1}{2}\vert \tilde uxhpxg\vert$, otherwise $rr$ cannot be a suffix of $\tilde uxhpxg$. Because $\vert h\vert>\vert p\vert > \max\{\vert \tilde ux\vert,\vert x\tilde v\vert\}$ we have that $r\in \Suffix(hpxg)$. Since $\occur(hp,x)=0$, $\vert h\vert>\vert p\vert>\vert x\tilde v\vert$, and $xg\in \Suffix(r)$ it follows that there is are words $h_1,h_2$ such that $\tilde uxhpxg=\tilde uxh_1h_2pxg$, $r=h_2pxg$ and $r\in \Suffix(\tilde uxh_1)$. It follows that $xg\in \Suffix(\tilde uxh_1)$ and because $\occur(h_1,x)=0$ we have that $\vert h_1\vert\leq \vert g\vert$. Since $\vert p\vert>\vert \tilde ux\vert$ we get that $\vert h_2pxg\vert>\vert\tilde uxg\vert\geq \vert \tilde uxh_1\vert$; hence $\vert r\vert>\vert\tilde uxh_1\vert$. This is a contradiction.
\end{itemize}
We conclude that there are no word $r$  and a no prefix $g\in \Prefix(\tilde v)$ such that with $rr\in \Suffix(\tilde uxhpxg)$.  Hence $\tilde uxhpx\tilde v\in \pfl_{k,\alpha}$. Due to the construction of $p$ and $h$ we have that $u\in \Prefix(\tilde uxhpx\tilde v)$ and $v\in\Suffix(\tilde uxhpx\tilde v)$. This completes the proof.
\end{proof}

\section*{Acknowledgments}
The author acknowledges support by the Czech Science
Foundation grant GA\v CR 13-03538S and by the Grant Agency of the Czech Technical University in Prague, grant No. SGS14/205/OHK4/3T/14.

\bibliographystyle{siam}
\IfFileExists{biblio.bib}{\bibliography{biblio}}{\bibliography{../!bibliography/biblio}}

\end{document}